\documentclass[12pt, letterpaper]{article}
\usepackage{graphicx}
\usepackage{amsfonts}
\usepackage{latexsym,amssymb,amsmath}
\usepackage{fullpage}
\usepackage[T1]{fontenc}

\usepackage{paralist}
\usepackage[bottom]{footmisc}
\usepackage[linesnumbered,lined,boxruled, boxed, vlined, noend]{algorithm2e}

\usepackage{color}
\usepackage{dsfont}
\usepackage{hyperref}

\hypersetup{colorlinks,allcolors=black}

\newcommand{\suc}{succ}
\newcommand{\pre}{pred}
\newcommand{\con}{orth}
\newcommand{\bconf}{borth}

\newtheorem{theorem}{Theorem}[section]
\newtheorem{lemma}[theorem]{Lemma}

\newenvironment{proof}[1][Proof]{\begin{trivlist}
\item[\hskip \labelsep {\bfseries #1}]}{\end{trivlist}}
\newenvironment{definition}[1][Definition]{\begin{trivlist}
\item[\hskip \labelsep {\bfseries #1}]}{\end{trivlist}}

\title{Computing the Optimal Longest Queue Length in Torus Networks}
\author{Mehrdad Alisagari\footnote{Department of Computer Science, California State University Long Beach: mehrdad.aliasgari@csulb.edu},
 Burkhard Englert\footnote{Department of Computer Science, University of North Caroline Wilmington: burkharde@gmail.com},
Oscar Morales-Ponce\footnote{Department of Computer Science, California State University Long Beach: oscar.morales-ponce@csulb.edu}}
\date{}

\begin{document}       

\maketitle

\begin{abstract}
A collection of $k$ mobile agents is arbitrarily deployed in the edges of a directed torus network where agents perpetually move to the successor edge. Each node has a switch that allows one agent of the two incoming edges to pass to its successor edge in every round. The goal is to obtain a switch scheduling to reach and maintain a configuration where the longest queue length is minimum. We consider a synchronous system. We use the concept of conflict graphs to model the local conflicts that occur with incident links. We show that there does not exist an algorithm that can reduce the number of agents in any conflict cycle of the conflict graph providing that all the links have at least 2 agents at every round. Hence, the lower bound is at least the average queue length of the conflict cycle with the maximum average queue length. Next, we present a centralized algorithm that computes a strategy in $O(n\log n)$ time for each round that attains the optimal queue length in $O(\sigma n)$ rounds where $n$ is the number of nodes in the network and $\sigma$ is the standard deviations of the queue lengths in the initial setting. Our technique is based on network flooding on conflict graphs. Next, we consider a distributed system where nodes have access to the length of their queues and use communication to self-coordinate with nearby nodes.  We present a local algorithm using only the information of the queue lengths at distance two. We show that the algorithm attains the optimal queue length in $O(\sigma C_{max}^2)$ rounds where $C_{max}$ is the length of the longest conflict cycle with the maximum average queue length.  
\end{abstract}

\section{Introduction}

Protecting a region with a collection of agents has recently attracted the attention of the distributed computing community.  The goal of theses problems is to understand the interaction of the agents, provide lower bounds that any algorithm can attain as well as providing strategies that guarantee optimal or approximation upper bounds. In this paper, we study a variant of the problem where agents are arbitrarily deployed in a graph. Agents perpetually move along a given cycle. However, agents on different cycles may collide after crossing conflicting nodes. To avoid collisions, a switch at each node coordinates the agents. The goal is to reach and maintain a configuration where the longest queue length is minimum. However,  the use of switches can result in unbalancing the number of agents in each link. Is it possible to provide a switch schedule that guarantees an even distribution of the agents? To our knowledge, there are no formal studies that provide lower and tight upper bounds. In this paper, we address a deterministic version of the problem in symmetric networks. Specifically, we study the problem in torus networks of dimension $\sqrt{n} \times \sqrt{n}$, i.e., the torus has $\sqrt{n}$ horizontal and $\sqrt{n}$ vertical cycles or rings. Initially, a collection of $k$ agents is arbitrarily deployed in the rings such that each ring has $k\sqrt{n}$ agents. To simplify the presentation, we refer to the links as queues to denote the order of the agents. We consider a synchronous system where each node is crossed by two agents in each round.  The natural question is if it is always possible to reach a deployment where all the links have the same number of agents, i.e., queue length $k$. We can answer the question by determining the minimum longest queue length in each link that can be reached. To reduce the queue length, the agent in front of the queue can cross to the successor link meanwhile, the agent in front of the predecessor queue does not cross. However, this may increase the queue length of its incident links. We call the links that are negatively affected by the scheduling conflict links. Further, the conflicts can extend into conflict paths, which in turn can form conflict cycles.  We model these conflicts using \emph{conflict graphs}.   

We show that there does not exist an algorithm that can reduce the number of agents in any conflict cycle if the links have at least 2 agents at  every round. Hence, the optimal lower bound is at least the maximum average length among all the conflict cycles in the network. Next, we present a centralized algorithm for directed torus networks that computes a strategy in $O(n\log n)$ time that minimizes the longest queue length in $O(\sigma n)$ rounds where $n$ is the size of the network and $\sigma$ is the standard deviation of the queue lengths in the initial deployment. The algorithm uses the conflict graph to apply a technique called flooding, used in routing algorithms, to minimize the longest queue length. Next, we present a local algorithm that uses the queue length of links at a distance two to minimize the longest queue length in at most $O(\sigma C_{max}^2)$ rounds where $C_{max}$ is the length of the longest conflict cycle with maximum average queue length in the initial deployment. We observe that our results can be applied to non-deterministic settings and general topologies.



\subsubsection{Organization of the paper}\label{sec:contribution}


We present the formal model and the problem statement in Section~\ref{sec:model}.  The related work is presented in Section~\ref{sec:relatedwork}. The concept of conflict graph is presented in Section~\ref{sec:traffic} and  the lower bound in Section~\ref{sec:lower}. In Section~\ref{sec:minmax}, we present the centralized algorithm meanwhile, the local algorithm is presented in Section~\ref{sec:localalgo}. We conclude the paper in Section~\ref{sec:conclusion}. Due to space constraints, some proofs are left.

\section{Model}\label{sec:model}

Let $G=(V,E)$ denote a graph where $V$ is the set of vertices (or nodes) and $E= \{\{u,v\} : u,v \in V\}$ is the set of links. An orientation of $G$ is the directed graph obtained by assigning to each link a direction.  Let  $(u,v) \in \overrightarrow{E}$ be the directed link with source at $u$ and destination at $v$. We consider a torus graph of dimension $m\times m$ which consists of $m$ horizontal and  $m$ vertical rings. More specifically,  a torus network of dimension $m \times m$ is the graph $G=(V, E)$ where  $V=\{v_{i,j} |  \forall i \in [0, m-1] \wedge \forall  j \in [0,m-1]\}$ is the set of vertices and $E =\{\{v_{i,j},v_{k, l}\} | (k = i+1 \bmod m \wedge l = j) \vee (k = i \wedge l = j+1 \bmod m) \}$ is the set of links. Thus, the number of vertices is $n= m^2$ and the number of links is $2n$. We assume that each link has a queue where agents wait. Throughout the paper, we refer to  the ring formed with links $\{v_{i,j},v_{i, (j+1) \bmod m}\}$ for all $j\in[0,m-1]$ as the horizontal ring $i$. Similarly, we refer to the ring formed with links $\{v_{i,j},v_{(i+1) \bmod m,j}\}$ for all $i\in[0,m-1]$ as the vertical ring $j$. Observe that each vertex $v_{i,j}$ has two incoming and two outgoing links. For each  $v$, let $\overrightarrow{E}^{in}(v)$ denote the set of incoming links to $v$ and $\overrightarrow{E}^{out}(v)$ denote the set of outgoing links from $v$. Given a link $e =(u,v)$, let $\suc(e)$, $\pre(e)$ denote the successor and predecessor link of $e$ as defined by the direction of the ring connecting $u$ and $v$. To simplify our presentation, we consider $m$ to be an even number and  a torus where the direction of the rings alternate. However, our solutions and analysis also work in tori where the rings have any arbitrarily direction.

Let  $\Omega$ be the set of agents arbitrarily deployed in the link queues of the torus. 
We divide the time in rounds and in every round,  a semaphore in each vertex $v$ allows two agents  crossing to  their successor links, i.e., let $e, e' \in \overrightarrow{E}^{in}(v)$, the semaphore either allows an agent in $e$ and $e'$ moving to $\suc(e)$ and $\suc(e')$, respectively, or two agents in $e$ moving to $\suc(e)$.  Observe that each agent infinitely often traverses its initial ring. Therefore,  the number of agents in each ring remains the same at all times.  We refer to the crossing schedule in each vertex as the green time assignment. Let $g_e(r)$ denote the green time of $e$ at round $r$. The green time assignment has the constraint  that $g_e(r) + g_{e'}(r) = 2$ where $e$ and $e'$ are in $\overrightarrow{E}^{in}(v)$. There are two ways to reduce the queue length of $e=(u,v)$ in one round:
\begin{compactenum}
\item Allowing two agents in front of $e$ moving to $\suc(e)$. However, the queue length  at the other incoming link incident to $v$, i.e.,  $(u', v) = \con(e)$ as well as the queue length of $\suc(e)$  increase.  Thus, $e$ is in forward conflict with $\suc(e)$  and  $\con(e)$.

\item Preventing the agents  in front of $\pre(e)$ to move to $e$. However, the queue length of the other outgoing link from $u$, i.e., $(u, w) =\bconf(e)$ as well as  the queue length of $\pre(e)$  increase. Thus, $e$ is in backward conflict with $\pre(e)$ and $\bconf(e)$.
\end{compactenum}

Let $w_e(r)$ be the number of agents in $e$ at round $r$. We define the agents deployment at time $r$ as $W(r) = \{ w_e(r) : \forall e \in E\}$. Throughout the paper, we usually omit the time reference if it is clear from the context.

\begin{definition} \label{def:dynamics}
Given an agent deployment $W(r) = \{ w_e(r) : \forall e \in E\}$ and a set $A(r) = \{ g_e(r) : \forall e \in E\}$ of green time assignments,  the deployment of $e \in E$ at round $r+1$ is given by:  

$$
w_e(r+1)  :=  w_e(r) - \min(g_e(r), w_e(r))   +  \min(g_{\pre(e)}(r), w_{\pre(e)}(r)).
$$

\end{definition}

We say that the  network is saturated if  $w_e(r)  \geq 2$ for all $e\in E$. Therefore, we can simplify the model in saturated torus networks and rewrite it as: $$w_e(r+1) :=  w_e(r) - g_e(r) + g_{\pre(e)}(r).$$

{\bf Problem Statement:}
Given an initial agent deployment $W$ in an oriented torus $G=(V, E)$, the overall objective is to determine the green time assignments such that the longest queue is minimum after a finite number of rounds. Let $\mathcal{A}(r)$ denote the set of green time assignments at round $r$.  Formally,

\bf{Problem:}
\normalfont 
Determine the green time assignments $\mathcal{A}(0), \mathcal{A}(1), ...,$ $\mathcal{A}(r), ...$ such that there exists a time $r$ when
the longest queue length overall $e \in E$ is  minimum thereafter.  Let $\phi$ denote the minimum longest queue length.

\section{Related Work}\label{sec:relatedwork}

\emph{Kortsarz} and \emph{Peleg}~\cite{kortsarz1994traffic} studied the traffic-light scheduling for route scheduling on a two-dimensional grid. 
They show upper bounds according to the number of directions. unlike our model where each vertex always allows one agent crossing (in 
either direction), edges can be activated and deactivated at any round. Moreover, in our setting agents infinitely often traverse the grid in one direction.
Our problem is closely related to the packet switching networks where switches can forward one packet at the time. The authors in \cite{davies1972control}  propose isarithmic networks where a constant number of agents (packets) in the overall network is maintained. When a packet arrives at its destination, it is replaced with new payload before putting back into the system. In our model, the number of agents remains constant at every time. Agents can be used to transmit data. Therefore, our model can be considered an isarithmic network and our results can be applied to these networks.

Our problem is also related to the routing problem in high-performance computers where cores are interconnected with different topologies such as torus networks.  Cores then communicate sending messages which are routed through the network.  Its performance and efficiency largely depend on routing techniques. Many authors have proposed different routing techniques that deal with the throughput, e.g.,~\cite{ebrahimi2011agent,ebrahimi2012mafa,farahnakian2011q,jog2013owl}. Despite providing optimal routes with respect with some metric, it is not clear how concurrent routes affect the system. In other words, how the system capacity is affected with a high volume of concurrent routes. Although we do not deal directly with routing, our problem provides the minimum longest queue that can be used as a primary metric for the system.

\section{Traffic Network Flows}\label{sec:traffic}

In this section, we introduce the concept of the conflict graph of an oriented torus $G=(V,\overrightarrow{E})$ that allows us applying flooding techniques to $G$ in a natural way. Next, we introduce the green time shifts that we use to define the shift assignments. Finally, we introduce conflict cycles in the conflict graphs.

Intuitively, in the conflict graph, the links of the torus are the set of vertices and two vertices $u,v$ are adjacent if reducing the queue length of $u$ in the direction of $v$ increases the queue length of $v$.

\begin{figure}[ht!]
\centering
\includegraphics[scale=.4]{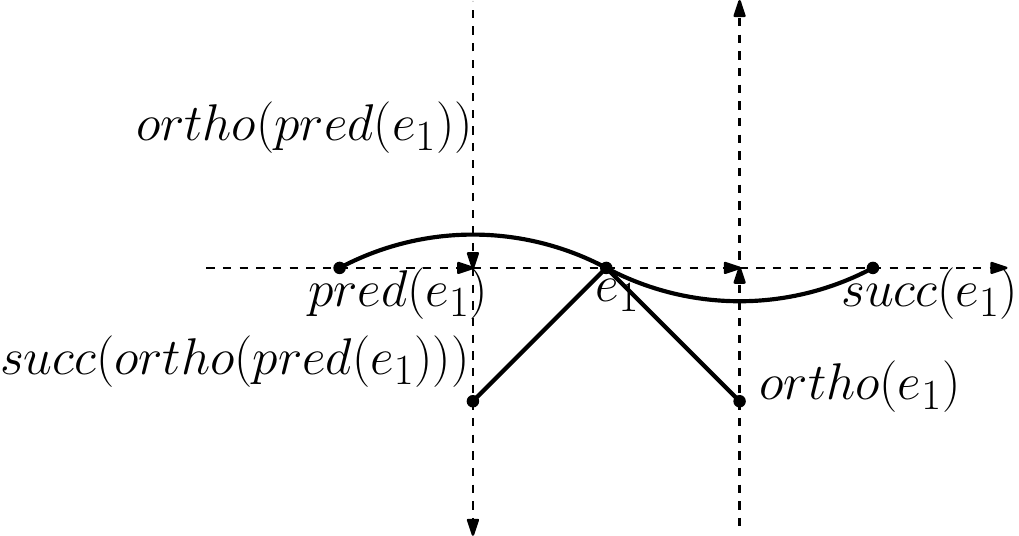}
\caption{Conflict links (Dotted lines denote the links of the torus and solid lines the links of the conflict the conflict graph.)}
\label{fig:conflictGraph}
\end{figure}

\begin{definition}\emph{(Conflict Graph)}\label{def:conflict}
Given an oriented  torus $G=(V,\overrightarrow{E})$ the conflict graph of $G$, is $L(G)=(E,\overrightarrow{E}')$; see Figure~\ref{fig:conflictGraph}, where for any $e_1, e_2 \in E$ there exists an edge $(e_1, e_2) \in \overrightarrow{E}'$ if and only if one of the following conditions holds:
\begin{compactenum}
\item $e_2 = \suc(e_1)$ (Equivalent $e_1 = \pre(e_2)$).
\item $\exists v  : \{e_1, e_2\} \subseteq \overrightarrow{E}^{in}(v)$; i.e., $e_2 = \con(e_1)$.
\item $\exists v  : \{e_1, e_2\} \subseteq \overrightarrow{E}^{out}(v)$; i.e., $e_2 = \bconf(e_1)))$.
\end{compactenum}
\end{definition}
 
The definition of conflict graphs can be  extended to other  digraphs.  For a fixed link $e \in \overrightarrow{E}$, we denote the  set of forward conflict neighbors of $e$ as $N^+(e) = \{\suc(e),  \con(e)\}$ and the set of backward conflict neighbors as $N^-(e) = \{\pre(e),  \bconf(e)\}$.  

For each link $e$, let $g_e  = [0, 2]$ be the green time shift subject to 
\begin{equation} \label{eq:constraint}
g_e + g_{\con(e)} = 2.
\end{equation}
Let $A(r)$ be the set of green time shifts.

\begin{definition}\emph{(Traffic Network)} \label{def:network}
Let $L(G)$ be the conflict graph  of $G=(V,\overrightarrow{E})$ with an agent deployment $W(r)$ at round $r$. A  network is defined by the tuple $TN(r) = (L(G), W(r), A(r))$. At the beginning of every round, $s_e(r) = 0$ for all $e \in \overrightarrow{E}$.
\end{definition}

Throughout the paper, we refer to the network as $TN = (L(G), W, A)$  if it is clear from the context.

For every link $e=(u,v) \in \overrightarrow{E}$, we can reduce its queue length by setting $g_e = 2$. However, from Constraint \ref{eq:constraint}, $g_{\con(e)} = 0$. We call this operation forward flow.  Similarly, we can reduce its queue length by setting $g_{\pre(e)} = 0$. However, from Constraint~\ref{eq:constraint}, $g_{\bconf(e)} = 2$. We call this operation backward flow. The following definition formalizes the flows.

\begin{definition}\emph{(Modified flows in Traffic Network)} \label{def:flow}
The forward flow  change operation $f^+_e$ sets  $g_e  = 2$ and $g_{\con(e)} = 0$. Analogously, the backward flow change operation $f^-_e$ sets  $g_{\pre(e)}  = 0 \mbox{ and }g_{\con(\pre(e))} = 2.$
\end{definition}

\begin{figure}[ht!]
\centering
\includegraphics[scale=0.6]{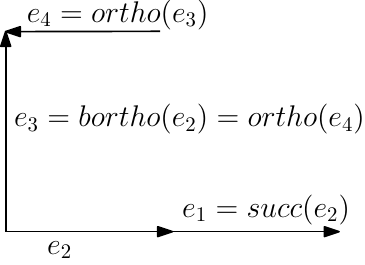}
\caption{A conflict path.}
\label{fig:conflictpath}
\end{figure}

Next, we introduce the concept of conflict path $P$ in $TN$; refer to Figure~\ref{fig:conflictpath}. We say that a path $P=\{e_{1}, e_{2},...,e_{l}\}$ in $TN$ is a conflict path if 
either  $|P| = 2$ and $e_1$ and $e_2$ are in conflict or  $|P| \geq 2$ and for  every three consecutive links $e_j,e_{j+1}, e_{j+2} \in P$, exactly one of the following statement follows:
\begin{compactenum}
\item   if $e_{j+1} = \con(e_j)$ or $e_{j+1} = \pre(e_j)$, \\then $e_{j+2} \in N^-(e_{j+1})$
\item   if $e_{j+1} = \suc(e_j)$ or $e_{j+1} = \bconf(e_j)$,\\then $e_{j+2} \in N^+(e_{j+1})$
\end{compactenum}

\begin{figure}[ht!]
\centering
\includegraphics[scale=0.9]{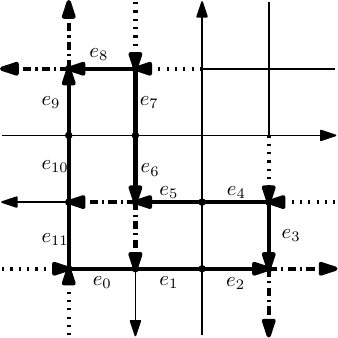}
\caption{A conflict cycle is represented with bold solid links, dotted arrows are entry links and dashed dotted arrows are exit links. }
\label{fig:conflictcycle}
\end{figure}


We define a conflict cycle $C$ of length $l$ to be a closed conflict path, i.e., $C=\{e_{0}, e_{1},...,e_{l-1}, e_{0}\}$; see Figure~\ref{fig:conflictcycle}.


We say that a  link  $e \in \overrightarrow{E}$  is an \emph{exit link} if $e$ is in the conflict cycle $C$ but $\suc(e)$ is not in $C$. Similarly,  a link $e \in \overrightarrow{E}$  is an \emph{entry link} if $e$ is not in $C$, but $\suc(e)$ is in $C$. A vertex $v \in C$  is an \emph{exit cycle} if both links in $\overrightarrow{E}^{out}(v)$ are exit links. A vertex $v \in C$  is an \emph{entry cycle} if both links in $\overrightarrow{E}^{in}(v)$ are entry links. Let $V^-(C)$ and $V^+(C)$ denote the set of exit cycles and the set of entry cycles in $C$, respectively.

\begin{lemma}\label{lem:equal}
In any conflict cycle $C$, the number of entry vertices is equal to the number of exit vertices, i.e., $|V^+(C)| = |V^-(C)|$.
\end{lemma}

\begin{proof}
Let $C=\{e_0, e_1, ..., e_{|C|-1}, e_0\}$. If $|V^-(C)| = 0$, then $C$ is either a vertical or a horizontal ring since for each two links $e_i, e_{i+1}$ in $C$, it holds that either $e_i = \suc(e_{i+1})$ or $e_i = \pre(e_{i+1})$. Therefore, $|V^+(C)|=0$.

Assume then that $|V^-(C)| > 0$. Let $R$ be the links of any ring $R$ of the torus. Let $P$ be any consecutive conflict path of $C$ such that $P$ is in $R$.  Let $u,v$ be the vertices at the endpoints of $P$. We show that $u$ is an entry vertex and $v$ is an exit vertex. Indeed, if that were not the case, $P$ would have distinct directions. Thus, it contradicts the fact that each ring has a different direction. The lemma follows since each section of $C$ in a ring is delimited by one entry vertex and one exit vertex. 
\end{proof}

 \section{Lower Bound}\label{sec:lower}

In this section, we present the lower bound on the minimum longest queue that any algorithm can attain in saturated networks. Recall that a saturated network is a $TN$ where $w_e(r) \geq 3$ for each $e \in TN$.

Observe that a conflict cycle is,  indeed, equivalent to a deadlock in the sense that reducing the number of agents in any queue of a cycle increases the queue length of another link in the cycle. 

\begin{lemma}
\label{lemma:constant}
Let $TN =(L(G), W, A)$ be a saturated  network and let $\mathcal{C}$ be the set of all conflict cycles in $L(G)$.  Then, the number of agents in each conflict cycle $C \in \mathcal{C}$ remains constant for all rounds.
\end{lemma}
\begin{proof}
Consider any green time assignment $A$. Since $TN$ is saturated, $w_e(r) \geq 3$ for all links $e \in TN$. Let us consider  any conflict cycle $C \in \mathcal{C}$. We show that the number of agents in $C$ remains constant in the next round, i.e., $$\sum_{e \in C} w_e(r) = \sum_{e \in C} w_e(r+1).$$ Let $P^i = e^i_1, e^i_{2}, ...,e^i_{l_i}$ be the segment of $l_i$ contiguous  links in the same ring of $C$. Observe that $C$ is formed by the ring segments $P^i$ that alternate between horizontal and vertical rings. Let us consider $P^i$ and let $W(P^i, r+1)$ be the number of agents at round $r+1$ in $P^i$. 
Then, 
$$
\begin{array}{rcl}
W(P^i, r+1) &=& \sum_{j=1}^{|P^i|} \left(w_{e^i_{j}} - g_{e^i_{j}} + g_{\pre(e^i_{j})} \right) \\
&= & \left(\sum_{j=1}^{|P^i|} w_{e^i_{j}}\right) + g_{\pre(e^i_1)} -  g_{e^i_{|P^i|}}.
\end{array}
$$
Thus, the number of agents in $P^i$ changes  only due to the entry and exit cycles.

We can calculate the number of agents in $C$ by asumming at least two agents in each link.  We show   $$\sum_{i}W(P^i, r) = \sum_{i}W(P^i, r+1).$$ Observe that $$\sum_{i}W(P^i, r+1)  = \sum_{i}\left(\left(\sum_{j=1}^{|P^i|} w_{e^i_{j}} \right)+ g_{\pre(e^i_1)} -  g_{e^i_{|P^i|}}\right),$$ consider $P^i$ and $P^{i+1}$. Two cases can occur:

\begin{compactitem}
\item $P^i$  and $P^{i+1}$ have incoming conflict links. W.l.o.g. assume that $e^{i+1}_{1} = \con(e^{i}_{|P^i|})$. Therefore, the number of  agents that leave the cycle $C$ is 
$g_{e^{i+1}_{1}}  +  g_{e^{i}_{|P^i|}} = 1.$

\item $P^i$ and $P^{i+1}$ have outgoing conflict links. W.l.o.g. assume that $e^{i+1}_{|P^{i+1}|} = \bconf(e^{i}_{1})$. Therefore, the number of  agents that enter the cycle $C$  is 
$g_{\pre({e^{i}_{|P^{i}|})}} + g_{\bconf(e^{i+1}_{1})} = 1.$

\end{compactitem}

Thus, the number of agents in $C$ is:
$$
\sum_{i}W(P^i, r+1)   = 
 \sum_{i}\sum_{l=1}^{|P^i|} w_{e^i_{l}} + |V^{in(C)}| - |V^{out(C)}|
$$
Further, from Lemma~\ref{lem:equal}, $|V^{in(C)}|  =  |V^{out(C)}|$. Therefore, $$\sum_{i}W(P^i, r) = \sum_{i}W(P^i, r+1).$$

\end{proof}

From Lemma~\ref{lemma:constant}, the best that any algorithm can do is to uniformly distribute the agents along the conflict cycle with the maximum average of agents as shown in the following theorem.

\begin{theorem}\label{thm:low1}
Let $TN =(L(G), W, S)$ be a saturated  network. Then the longest queue length $\phi$ that any greedy algorithm can attain is at least $$\phi \geq \max _{C \in \mathcal{C}} \left(\left \lceil \frac{{\sum_{e \in C}}w_e(r) }{|C|}\right \rceil  \right)$$ where $\mathcal{C}$ is the set of all conflict cycles in $L(G)$.
\end{theorem}

\begin{proof} 
The theorem follows from Lemma~\ref{lemma:constant} and the pigeon hole principle.
\end{proof}

In the sequel, we say that $C$ is a critical conflict cycle if $\phi  = \left\lceil \frac{\sum_{ e \in C} w_e}{|C|}  \right\rceil$.

\section{Minimizing the Longest Queue Length in Saturated Traffic Networks} \label{sec:minmax}

In this section, we present a strategy that minimizes the maximum queue length in saturated networks. The main idea of the algorithm is based on flooding in distributed computing using Definition~\ref{def:flow}. We start the flooding in a queue with the longest queue length $\phi'$ which provokes that the conflict links get flooded. We continue flooding the conflict links until all the queue lengths in the flooding paths are less than $\phi'$. 

Given a network $TN= (L(G), W, A)$, an  $i$-conflict path $e_1,e_2, ..., e_a$ of $TN$ is  a conflict path  where $w_{e_1} = i$ and  $w_{e_j} = i-1$ for all $j \in [2, a]$. An $i$-conflict path $e=e_1,e_2, ..., e_a$ is \emph{forward} $i$-conflict path if $e_2 \in N^+(e)$. Otherwise, is a \emph{backward} $i$-conflict path.  We say that $e'$ is the \emph{forward $i$-conflict adjacent}  of $e$ if there exists a conflict path $e=e_1,e_2, .., e_a =e'$ such that  $w_{e_1} = w_{e_a} = i$ and $w_{e_j} =  i-1$ for $j \in [2,a-1]$; see Figure~\ref{fig:conflictadjacent}. Analogously, we define the \emph{backward $i$-conflict adjacent}. Observe that it is always possible to reduce the queue length by flooding the $i$-conflict path. On the other hand, the path between two links that are $i$-conflict adjacent cannot be reduced.

\begin{figure}[ht!]
\centering
\includegraphics[scale=0.9]{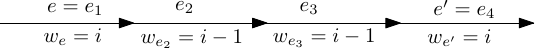}
\caption{$e'$ is the forward $i$-conflict adjacent of $e$ and 
$e$ is the backward $i$-conflict adjacent of $e'$ }
\label{fig:conflictadjacent}
\end{figure}

Consider a link $e$, the forward indicative function returns false if $e$ does not have a forward $i$-conflict adjacent neighbor and true if there exists $e'$ such $e'$ is the forward $i$-conflict adjacent of $e$. Formally,
$$
\mathds{1}_e^{+} = 
\begin{cases}
true & \mbox{ if } \exists e' \mbox{ such that $e$ and $e'$ are forward  $i$-conflict adjacent} \\
false & \mbox{otherwise }
\end{cases}
$$ 

Similarly we define for backward $i$-conflict adjacent, i.e.,  $\mathds{1}_e^{-}$.

Our approach consists of two steps. First, we present an efficient algorithm to find a link with maximum queue length such that the value of either $\mathds{1}_e^{+}$ or  $\mathds{1}_e^{-}$ is $false$. Once we identify a link $e$ such that either  $\mathds{1}_e^{+}$ or $\mathds{1}_e^{-}$ is $false$ we reduce its queue length. These two steps are repeated until there is no link $e$ with maximum queue length such that neither $\mathds{1}_e^{+}$ nor $\mathds{1}_e^{-}$ is $false$.  

A conflict tree $T^{+}_e(r)$ is a tree of forward $i$-conflict paths without cycles at round $r$ where $w_e = i$ and $w_{e'} = w_e-1$ for each $e' \in T^{+}_e(r)$ different from $e$. Similar, we define $T^{-}_e(r)$. $T^{dir}_e(r)$  is maximal if it cannot be extended in the direction $dir$.

Next, we present an algorithm that given  a link $e$ with the maximum queue length it determines if either $\mathds{1}_e^{+}$  or $\mathds{1}_e^{-}$ is $false$. Observe that the naive algorithm for determining the values $\mathds{1}_e^{+}$ and $\mathds{1}_e^{-}$ takes quadratic time as every conflict path is traversed and potentially there exist a quadratic number of conflict paths. Instead, we construct $T^{dir}_e(r)$  until $T^{+}_e(r)$ is maximal or a conflict cycle appears. 

\begin{lemma}\label{lem:adjacent}
Given $e$ and $dir$ such that $w_e$ is the largest queue length and $dir \in \{+, -\}$, there exists a linear time algorithm that determines whether $\mathds{1}_{e}^{+} = false$ or $e$ is part of an $i$-conflict cycle.
\end{lemma}

\begin{proof}
Let $i$ be the length of the longest queue length. Since $w_e$ is the largest queue length, $w_e=i$. Observe that  if there exists a link $e'$ such that $w_{e'} =i$ and  $e'\in N^+(e)$, then $\mathds{1}_{e}^{+}$ is $true$ and either $\mathds{1}_{e'}^{+}$ or $\mathds{1}_{e'}^{-}$  is $true$. However, if $w_{e'} < i$, they form an $i$-conflict path. Therefore,  we need to check the paths and potentially traverse many $i$-conflict paths which would increase the complexity. To keep track of the already visited $i$-conflict paths, we color  $(e,dir)$ with black if  it has not been visited and with red if it has been visited where  $dir \in (+, -)$.  We use a stack $Q$ to store the links with queue length $i$ that need to be visited in the $i$-conflict paths. The idea is to explore the  $i$-conflict trees  until the tree is maximal. Thus, we can determine if either $\mathds{1}_{e}^{dir}$ is $true$ or $false$. Let $T^{dir}_e(r)$ be the conflict tree rooted at $e$.  The proof is by induction on $r$.

Consider for the basic step $r=0$.  In other words $e$ is the only element in $T^{dir}_e(0)$. Suppose that $dir = +$. We color $(e, +)$ with red and if $w_{\suc(e)}  < i -1$ and $w_{\con(e)} < i-1$, then  $e$ is not conflict adjacent with neither $\suc(e)$ nor $\con(e)$. Therefore, $\mathds{1}_{e}^{+}$ is $false$. On the other hand, if either $w_{\suc(e)}  = i$ or $w_{\con(e)}  = i$, then $e$ is  $i$-conflict adjacent with either $\suc(e)$ or $\con(e)$ or both. Therefore,  $\mathds{1}_{e}^{+}$ is $true$ and if $w_{\suc(e)}  = i$, then  $\mathds{1}_{\suc(e)}^{-}$  is $true$ and if $w_{\con(e)}  = i$, then  $\mathds{1}_{\con(e)}^{+}$ is $true$. However, if either $w_{\suc(e)}  = i-1 $ or $w_{\con(e)} = i-1$ or both, then we need to explore the branches to determine whether $e$ is $i$-conflict adjacent of some link. Therefore, if $w_{\suc(e)}  = i-1$, we insert $(\suc(e), +)$ into the stack $Q$ to explore the forward $i$-conflict path. Similarly, if $w_{\con(e)}  = i-1 $, we insert $(\con(e), -)$ into the stack $Q$  to explore the backward $i$-conflict path. Analogously, we can show for  $dir = -$.

Suppose that the  inductive hypothesis holds for $T^{dir}_e(r)$ where $r>0$. Let $(e', dir)$ be the link and direction in the top of the stack $Q$ If $color(e', dir)$ is red, then from the maximality of  $T^{dir}_e(r)$  an $i$-conflict cycle has arisen. Otherwise, we continue exploring all the $i$-conflict adjacent links; refer to Figure~\ref{fig:tree}. Thus,  if $color(e', dir)$ is black we consider two cases.

\begin{figure}[ht!]
\centering
\includegraphics[scale=1]{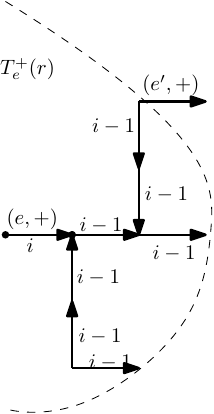}
\caption{Processing $(e', +)$ in the three $T^{+}_e(r)$.}
\label{fig:tree}
\end{figure}

\begin{compactitem}
\item $dir  = -$. Therefore, there exists an $i$-conflict path from $e$ to $e'$. We color $(e', -)$  with red. If $w_{\pre(e')}  <1-1$ and $w_{\bconf(e')} < i-1$, then  $e$ is not $i$-conflict adjacent with neither $\pre(e')$ nor $\bconf(e')$. However, if either $w_{\pre(e')}  = i$ or $w_{\bconf(e')}  = i$, then $e$ is  $i$-conflict adjacent with either $\pre(e')$ or $\bconf(e')$ or both. Therefore, $\mathds{1}_{e}^{-} = true$. If $w_{\pre(e')}  = i$,  then $\mathds{1}_{\pre(e')}^{+}$  is $true$ and if $w_{\bconf(e')}  = i$, then $\mathds{1}_{\bconf(e')}^{-}$ is $true$. However, if either $w_{\pre(e')}  = i-1$ or $w_{\bconf(e')} = i-1$ or both, then we need to explore the branches to determine whether $e$ is conflict adjacent with some link. Therefore, if $w_{\pre(e')}  = i-1$, we insert $(\pre(e'), -)$ into the stack $Q$ to explore the backward $i$-conflict path. Similarly, if $w_{\bconf(e')}  = i-1 $, we insert $(\bconf(e'), +)$ into the stack $Q$  to explore the forward $i$-conflict path.

\item $dir = +$.  The proof is similar to the previous case and is left as an exercise to the reader.

\end{compactitem}

If $Q$ is empty and there are no conflict cycles, then $\mathds{1}_{e}^{dir}$ is $false$. Observe that $T^{dir}_e(r)$ is maximal. The lemma follows.
\end{proof}

\begin{algorithm}[ht!]
  \caption{AdjacentConflict:  }
  \label{alg:conflicttree}
  \SetKwInput{InOut}{Input/Output}        
  \KwIn{$root$: such that the queue length is maximum }
        \KwIn{$dir$: $\{+, -\}$}        
        $i \gets w_{root}$\;
        Let $Q$ be a stack\;
        $push(Q, (root,dir))$\;
        \While{$Q  \not= \emptyset$ }
        {
            $(e, dir) \gets pop(Q)$\;
\tcc{A link is red if it has breen visted}
            \If{$color((e,dir)) = red$}    
            {
                \Return There is a $w_{root}$-conflict cycle\;
            }
            $color((e,  dir)) = red$\;
            \If{$dir = +$}
            {    
                \If {$w_{\suc(e)} = i$}              
                {    
                    $\mathds{1}_{\suc(e)}^{-} = true$\; 
                }
                \ElseIf {$w_{\suc(e)} = i -1 $}              
                {    
                    $push(Q , (\suc(e), +))$
                }
                \If {$w_{\con(e)} = i$}              
                {    
                    $\mathds{1}_{\con(e)}^{+} = true$\; 
                }
                \ElseIf {$w_{\con(e)} = i-1$}              
                {
                    $color(Q , (\con(e), -))$
                }        
            }            
            \Else
            {   
                \If {$w_{\pre(e)} =  i$}  
                {    
                    $\mathds{1}_{\con(e)}^{+} = true$\; 
                }
                \ElseIf {$w_{\pre(e)}  = i -1$}  
                {        
                    $push(Q , (\pre(e), -))$
                }
                \If {$w_{\bconf(e)} =  i$}  
                {    
                    $\mathds{1}_{\bconf(e)}^{-} = true$\; 

                }
                \ElseIf {$w_{\bconf(e)} = i -1$}  
                {
                     $push(Q , (\bconf(e), +))$
                }
            }

    }
    \Return $false$\;
\end{algorithm}

Algorithm~\ref{alg:conflicttree} can be used to find the value of $\mathds{1}_{e}^{dir}$  for every link with the longest queue length.

\begin{lemma}\label{lem:adjacent1}
Given a  network, we can compute  $\mathds{1}_{e}^{+}$ and $\mathds{1}_{e}^{-}$  for each link $e$  with the maximum queue length in linear time.
 \end{lemma}

\begin{proof}
Let $\mathcal{E}$ be the set of links with the maximum queue length. Initially, $(e,+)$ and $(e,-)$ are  black for all  links. For each $e \in \mathcal{E}$, if $color(e,+)$ is black, then we run Algorithm~\ref{alg:conflicttree} to determine $\mathds{1}_{e}^{+}$. Similarly, if $color(e,-)$ is black, then we run Algorithm~\ref{alg:conflicttree} to  determine $\mathds{1}_{e}^{-}$. Regarding the complexity, observe that every $i$-conflict path is traversed once. Therefore, the time complexity is linear on the number of edges. Since the graph is planar, the time complexity is $O(n)$ where $n$ is the number of intersections.
\end{proof}

We show in the next lemma how the queue length of a link  $e$ with maximum queue length such $\mathds{1}_{e}^{-} = false$ can be reduced. The main observation is that when $\mathds{1}_e^{dir}= false$, there are no $i$-conflict adjacent links of $e$ in the maximal tree $T^{dir}_e(0)$. Let $d(T^{dir}_e(0))$ be the length of the longest $i$-conflict path in $T^{dir}_e(0)$. 

\begin{theorem}\label{lem:flood}
Let $\omega(r) =   \sum_{ e  } |w_e(r) - \phi|$. There exists an algorithm that computes a green time assignment such that $\omega(r+1)  < \omega(r)$ in $O(n \log n)$ time. Further, it attains the optimal value in $O(\sigma n)$ rounds where $\sigma$ is the standard deviation of the queue lengths in the initial setting. 
\end{theorem}

\begin{proof}
Initially, $g_e(r) = 0$ for each link in $E$. Let $\mathcal{E}  = \{ e: \mathds{1}_{e}^{+} = false \mbox{ \bf or } \mathds{1}_{e}^{-} = false\}$. Consider  $e \in \mathcal{E}$ such that $w_e$ is maximum. If $\mathds{1}_{e}^{+} = false$, apply a forward flow $f^+_e$, and $g_e$  increases  and $g_{\con(e)}$ decreases by one. If $\mathds{1}_{e}^{-} = false$, apply a forward flow $f^-_e$, and $g_{\pre}$  decreases  and $g_{\con(\pre(e))}$ increases by one. Consider an incident link $e'$. 
\begin{compactitem}
\item If $e' \in \{ \suc(e), \bconf(e)\}$, $w_{e'}(r) \geq \max(w_{\suc(e')}(r), w_{\con(e')}(r))$ and  $g_{e'} =0$,  we apply a forward flow, i.e., $f^+_{e'}$. 
\item If $e' \in \{\con(e),  \pre(e)\}$, $w_{e'}(r) \geq \max(w_{\pre(e')}(r), w_{\bconf(e')}(r))$ and  $g_{\pre(e')} =0$,  we apply a forward flow, i.e., $f^-_{e'}$. 
\end{compactitem}

Inductively, the process continues until it cannot be extended. Otherwise, another link in $M$ is considered until no more progress can be done. The algorithm is presented in Algorithm~\ref{alg:flooding}. 

If $\mathcal{E}$ is empty, then the current longest queue length is minimum. Let  $\phi$ be the optimal queue length. Observe that the  maximum number of rounds is $\sum_{ e : w_e > \phi } w_e - \phi  \leq \sum_{ e : w_e > \mu }  w_e - k$  since  $\phi \geq  k$.  By Chebyshev's inequality, the maximum number of rounds is $an\sigma$  where $\sigma$  is the standard deviation and $a$ is a constant integer. Therefore, the number of rounds is $O(n \sigma )$.

\begin{algorithm}[ht!]
  \caption{Flooding}
  \label{alg:flooding}
  \SetKwInput{InOut}{Input/Output}        
      Let $Q$ be a stack\;    
      Let $\mathcal{E} = \{ e: \mathds{1}_{e}^{+} = false \mbox{ \bf or } \mathds{1}_{e}^{-} = false\}$\;
    \While{$M  \not= \emptyset$ }
    {
        Let $e \in \mathcal{E}$ such that $w_e$ is maximum\;
        \If{$\mathds{1}_{e}^{+}  =false$ }
        {    
            $push(Q , (e, +))$\;
        }
        \ElseIf{$\mathds{1}_{e}^{-}  =false$ }
        {    
            $push(Q , (e, -))$\;
        }
        
        \While{$Q  \not= \emptyset$ }
        {\label{alg:pop}
            $(e, dir) \gets pop(Q)$\;
            \If{$dir = `+`$ and $s_e = 0$ and $w_e \geq \max(w_{\suc(e)}, w_{\con(e)})$}
            {    
                $f^+_e$\;
                $push(Q , (\suc(e), `+`))$\;
                $push(Q , (\con(e), `-`))$\;
            }            
            \ElseIf{$dir = `-`$ and $s_{\pre(e)} = 0$ and $w_e \geq \max(w_{\pre(e)}, w_{\bconf(e)})$}
            {   
                $f^-_e$\;
                $push(Q , (\pre(e), `-`))$\;
                 $push(Q , (\bconf(e),`+`))$\;
            }
         }
    }    
\end{algorithm}

Regarding the complexity, we can compute $\mathcal{E}$ as follows. First, it sorts the links according to the queue length in $O(n\log n)$ time and stores them in an order set. We compute $\mathds{1}_{e}^{+}$  and  $\mathds{1}_{e}^{-}$ starting from the larger until all links have been visited. Therefore, since each link is visited once, we can compute $\mathcal{E}$ in $O(n \log n)$ time. To compute the green time assignment, we visit each link once. Therefore, the running time algorithm takes $O(n \log n)$. The theorem follows. 
\end{proof}

%




\section{Local Algorithm}\label{sec:localalgo}
In this section, we show that surprisingly the global problem can be solved locally. In other words, each queue concurrently makes a decision of whether sending a forward or backward flow based only on the queue length of the links at a distance at most two. To break symmetries, we assume that each link has a unique id. The idea is that links at the end of $i$-conflict path can be reduced safely. However, $i$-conflict adjacent links can prevent reaching the optimal queue length. We are able to reduce the queue length by keeping the flow direction when links can be part of $i$-conflict paths. 

For the proof we introduce some notation. Let $\max^+_e(r)$ denote the maximum value of the forward neighbors of $e$ and $\max^-_e(r)$ denote the maximum value of the backward neighbors of $e$. A critical link $e$ is a link with maximum queue length such that either $w_e(r) -1 = \max^+_e(r)$ or $w_e(r) -1 = \max^-_e(r)$. Let $m$ be the length of the maximum queue length. Let $e, e', e''$ be a conflict path. We say that a critical queue is rotating around its conflict cycle if either $w_{e'}(r-1) = m-1, w_{e''}(r-1) = m$, $w_e(r) = m-1, w_{e'}(r) = m, w_{e''}(r) = m-1$ and $w_e(r+1) = m, w_{e'}(r) = m-1$  (see Figure~\ref{fig:rotating} left) or $w_e(r-1) = m, w_{e'}(r-1) = m-1, w_{e}(r) = m-1, w_{e'}(r) = m, w_{e''}(r) = m-1$ and  $w_{e'}(r+1) = m-1, w_{e''}(r) = m$  (see Figure~\ref{fig:rotating} right).
 
\begin{figure}[ht!]
\centering
\includegraphics[scale=.56]{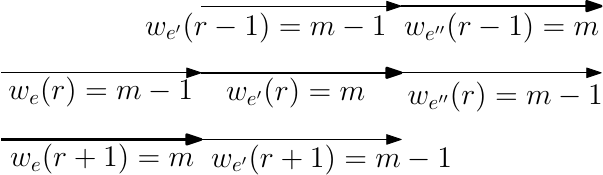} ~~~
\includegraphics[scale=.56]{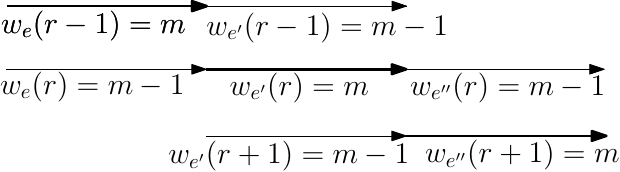}
\caption{Critical queue rotating left and right, respectively. }
\label{fig:rotating}
\end{figure}

Let $w_e  \succ^2 \max^+_e$ if  $w_e  > \max^+_e$ and either $w_e  > \max(\max^+_{\suc(e)}, \max^-_{\con(e)}) \mbox{ or }$,  $w_e  = \max(\max^+_{\suc(e)}, \max^-_{\con(e)})$ and $id_e  > id_{e'}$ for all $e' \in N^+(\suc(e)) \cup N^-(\con(e))$ where $w_{e'} = w_e$. Similarly, we define  $w_e  \succ^2 \max^-_e$.

\begin{algorithm}[ht!]
  \caption{Local Algorithm (Each link $e$ does in each round)}
  \label{alg:localalgorithm}
           \lIf{$w_e(r) -2 > \max^+_e(r) \mbox{ \bf or }    (w_e(r) -2 = \max^+_e(r) \mbox{ \bf and }  w_e(r)  \succ^2 \max^+_e(r))$ $   \mbox{ \bf or } (w_e(r) -1 = \max^+_e(r) \mbox{ \bf and }    w_e(r)  \succ^2 \max^+_e(r)    \mbox{ \bf and }  s_{(\pre(e))}(r-1) \geq 0)$}
           {
        $f^+$
    }
     \lElseIf{$w_e(r)  -2 > \max^-_e(r) \mbox{ \bf or }    (w_e(r)  -2 = \max^-_e(r) \mbox{ \bf and }  w_e(r)   \succ^2 \max^-_e(r))$ $   \mbox{ \bf or } (w_e(r)  -1 = \max-+_e(r) \mbox{ \bf and }    w_e(r)  \succ^2 \max^-_e(r)    \mbox{ \bf and }  s_{e}(r-1) \leq 0)$}
    {
        $f^-$
    }
       
\end{algorithm}

\begin{theorem}\label{thm:localalgo}
Algorithm~\ref{alg:localalgorithm} minimizes the longest queue length in at most $O(\sigma C_{max}^2)$ rounds where $\sigma$ is the standard deviation of the queue lengths in the initial setting and $C_{max}$ is the length of the longest conflict cycle with the maximum average queue length.
\end{theorem}

\begin{proof}
Let  $n_{i}(r)$ be the number of links with  rank $i$ (from maximum to minimum) and  $m_i(r)$ be the length of the queue with rank $i$ at round $r$. We show that in every step the algorithm makes progress. More specifically, we show that in every round either $n_{i}(r) < n_{i}(r+1)$  or $m_1(r+1) < m_1(r)$ or the longest queue length rotates in a consistent direction around its conflict cycles. 

Consider any link $e$ such that $w_e(r) >  \max^+_e(r)$. Let $e_i, e_{i+1}, e_{i+2}$ be the forward conflict path such that:
\begin{compactenum}
\item If either $w_{\suc(e_i)} > w_{\con(e_i)}$ or  $w_{\suc(e_i)} = w_{\con(e_i)}$ and $\max^+_{\suc(e_i)} > \max^-_{\con(e_i)}$, then $e_{i+1} = \suc(e_i)$, $e_{i+2} \in N^+(\suc(e_i))$  where $w_{e_{i+2}} = \max^+(\suc(e_i))$.

\item If either $w_{\suc(e_i)} < w_{\con(e_i)}$ or  $w_{\suc(e_i)} = w_{\con(e_i)}$ and $\max^+_{\suc(e_i)} < \max^-_{\con(e_i)}$, then $e_{i+1} = \con(e_i)$, $e_{i+2} \in N^-(\con(e_i))$  where $w_{e_{i+2}} = \max^-(\con(e_i))$.

\item If  $w_{\suc(e_i)} = w_{\con(e_i)}$ and $\max^+_{\suc(e_i)} = \max^-_{\con(e_i)}$, then $e_{i+2}$ is the link with maximum id and longest queue length, and if $e_{i+2} \in N^+(\suc(e_i))$, then $e_{i+1} = \suc(e_i)$, otherwise $e_{i+1} = \con(e_i)$ 
\end{compactenum}

Similarly, we define $e_i, e_{i-1}, e_{i-2}$ to be the backward conflict path. Consider a conflict path $e_{i-2}, e_{i-1}, e_i, e_{i+1}, e_{i+2}$ such that $w_{e_i}(r)$ is maximum, i.e., equal to $m_1(r)$ and either $w_{e_{i-1}}(r) < m_1(r)$ or $w_{e_{i+1}}(r) < m_1(r)$. Suppose without loss of generality that $e_{i-1} \in N^-(e_i)$ and $e_{i+1} \in N^+(e_i)$. Since $w_{e_i}(r)$ is maximum, no neighboring link sends a flow to $e_i$.  

\begin{compactitem}
\item $w_{e_i}(r) - 2 > w_{e_{i+1}}(r)$. Then, $e_i$  applies a forward shift, i.e.,  $f^+_{e_i}$. Therefore, $w_{e_i}(r+1) < m_1(r)$, Further, $w_{e_{i+1}}(r+1) < m_1(r)$ and
if $n_1(r) > 1$,  then $n_1(r+1) < n_1(r)$. 

\item $w_{e_i}(r) - 2 > w_{e_{i-1}}(r)$. Then, ${e_i}$  applies a backward shift, i.e., $f^-_{e_i}$. Therefore, $w_{e_i}(r+1) < m_1(r)$, Further, $w_{e_{i-1}}(r+1) < m_1(r)$ 
and if $n_1(r) > 1$,  then $n_1(r+1) < n_1(r)$. 

\item $w_{e_i}(r) - 2 = w_{e_{i+1}}(r)$.  If either $w_{e_i}(r) > w_{e_{i+2}}(r)$ or  $w_{e_i}(r) = w_{e_{i+2}}(r)$ and $id_{e_i} > id_{e_{i+2}}$, then $e_i$  applies a  forward shift, i.e., $f^+_{e_i}$. However, $e_{i+2}$ does not send a flow to $e_{i+1}$ since $id_{e_i} > id_{e_{i+2}}$. Therefore, $w_{e_{i+1}}(r+1) < m_1(r)$ and $w_{e_i}(r+1) = m_1(r)  - 1$ and if $n_1(r) > 1$,  then $n_1(r+1) < n_1(r)$. 

\item $w_{e_i}(r) - 2 = w_{e_{i-1}}(r)$.  If either $w_{e_i}(r) > w_{e_{i-2}}(r)$ or  $w_{e_i}(r) = w_{e_{i-2}}(r)$ and $id_{e_i} > id_{e_{i-2}}$, then ${e_i}$  applies a  backward shift, i.e.,  $f^-_{e_i}$. However,  $e_{i-2}$ does  not send a flow to $e_{i-1}$ since $id_{e_i} > id_{e_{i-2}}$. Therefore, $w_{e_{i-1}}(r+1) < m_1(r)$ and $w_{e_i}(r+1) = m_1(r)  - 1$ and if $n_1(r) > 1$,  then $n_1(r+1) < n_1(r)$. 

\item $w_{e_i}(r) - 1 = w_{e_{i+1}}(r)$ and $w_{e_i}(r) - 1 = w_{e_{i-1}}(r)$. Consider the previous round.  
\begin{compactitem}
\item $w_{e_{i-1}}(r-1) = m_1(r)-1$, $w_{e_i}(r-1) = m_1(r)$ and $w_{e_{i+1}}(r-1) = m_1(r)-1$. If $w_{e_i}(r)  \succ^2 \max^+_{e_i}(r)$, then $e$  applies a  forward shift. However,  $e_{i+2}$ does  not send a flow to $e_{i+1}$ since $id_{e_i} > id_{e_{i+2}}$. Otherwise, if $w_{e_i}(r)  \succ^2 \max^-_{e_i}(r)$, then $e$  applies a backward shift. However,  $e_{i-2}$ does  not send a flow to $e_{i-1}$ since $id_{e_i} > id_{e_{i-2}}$. Observe that the critical start rotating.

\item $w_{e_{i-1}}(r-1) = m_1(r)$, $w_{e_i}(r-1) = m_1(r)-1$ and $w_{e_{i+1}}(r-1) = m_1(r)-1$. Therefore, $s_{\pre(e_i)}(r-1) >0$. If $w_{e_i}(r)  \succ^2 \max^+_{e_i}(r)$, then $e_i$  applies a  forward shift  and $e_{i+2}$ does not send a flow to $e_{i+1}$. Observe that critical queue maintains a consistent rotation.

\item $w_{e_{i-1}}(r-1) = m_1(r)-1$, $w_{e_i}(r-1) = m_1(r)-1$ and $w_{e_{i+1}}(r-1) = m_1(r)$. Therefore, $s_{e_i}(r-1) < 0$. If $w_{e_i}(r)  \succ^2 \max^-_{e_i}(r)$, then $e_i$  applies a  backward shift  and $e_{i-2}$ does not send a flow to $e_{i-1}$. Observe that critical queue maintains a consistent rotation.
\end{compactitem}

\item Either, $w_{e_{i-1}}(r) = m_1(r) - 1$ and $w_{e_{i+1}}(r) = m_1(r)$ or  $w_{e_{i-1}}(r) = m_1(r)$ and $w_{e_{i+1}}(r) = m_1(r)-1$. Observe that if  $m_1(r)$ is not optimal, there must exists $e_i$ such that either $w_{e_i}(r)   \succ^2 \max^+_{e_i}(r)$ or $w_{e_i}(r)   \succ^2 \max^-_{e_i}(r)$ otherwise $n_1(r)$ is optimal. Thus, $e_i$ applies a forward or backward shift, respectively.
 
\end{compactitem}

Let $\phi$ be the optimal queue length and let $C_{max}$ the length of the longest conflict cycle. Observe that in the worst case, the longest queue length moves along  the $m_2$-conflict path before reducing its queue length by one. Since there are at most $n_{2}$ links and all the links are concurrent, it takes at most $(m_1-m_2)(n_1 + n_{2})$ rounds to reduce  $n_{1}$ links from $m_1$ to $m_2$ which results in at most  $n_{1} + n_{2}$ links with queue length $m_2$. Thus, the number of rounds is at most:

$$
\begin{array}{rcl}
(m_1- m_2)(n_{1} + n_{2}) & + &\\
(m_2 - m_3)(n_{1} +n_{2} + n_{3}) &+ & \\
 ...  &+ &\\
(m_{l-1} - m_l)(n_{1} + n_{2}  + n_{3} +...  + n_{l-1} + n_{l})  & \leq &  \\
\sum_{j =1}^{l} (m_j- m_{j+1})\sum_{i =1}^{j} n_{i} &  & 
\end{array}
$$

From Chebyshev's inequality,  $m_i - m_{i+1} \leq a \sigma$ for a constant $a$ where $\sigma$ is the standard deviation of the queue lengths of the initial  setting. Further  $n_{i} + n_{i+1}  \leq 2\max(n_{i}, n_{i+1})$ and  $l \leq \log_2 C_{max}$ since $n_{1} + n_{2}  + n_{3} +...  + n_{l-1} + n_{l} \leq C_{max}$. Therefore, $$\sum_{j =1}^{l} (m_j- m_{j+1})\sum_{i =1}^{j} n_{i}  \leq aC_{max}\sigma \sum_{i=1}^{\log_2 n} 2^i = O(\sigma C_{max}^2 ).$$ The theorem follows.
 
\end{proof}

\section{Conclusion}\label{sec:conclusion}

In this paper, we have studied the problem of determining the minimum longest queue length in torus networks.  We focus on the fundamental questions of lower bounds and upper bounds on the queue lengths when the agents move in a straight line. We present a global and a local algorithm. However, the number of rounds needed to reach the minimum remains an open problem. Another open problem is to consider asynchronous switches. This paper presents a new perspective on the problem of traffic lights. One direction is to study through simulation the effects when the model becomes probabilistic and, for example, includes turning, rates and agents that appear and disappear.

\bibliographystyle{splncs04}
\bibliography{ref}

\end{document}